\documentclass[12pt]{amsart}
\usepackage{blindtext}

\usepackage[utf8]{inputenc}
\usepackage[T2A,T1]{fontenc}
\usepackage[russian,english]{babel}

\usepackage{geometry}

\usepackage[utf8]{inputenc}

\usepackage{lipsum}
\newenvironment{acknowledgements}{
  
  \begin{abstract}
}{
  \end{abstract}
}

\usepackage{xcolor}
\usepackage[all]{xy}
\usepackage{amsmath}
\usepackage{amsfonts}
\usepackage{hyperref}
\usepackage{amssymb}
\usepackage{amscd}
\usepackage{amsthm}
\usepackage{latexsym}
\usepackage{amsbsy}
\usepackage{graphicx}
\usepackage{epstopdf}

\usepackage{shuffle}
\usepackage{mathtools}

\usepackage[dvipsnames]{xcolor}

\input xypic
\usepackage{tikz}
\usetikzlibrary{babel}
\usepackage{tikz-cd}

\usepackage{color}

\newtheorem{thm}{Theorem}[section]
\newtheorem{prop}[thm]{Proposition}

\newtheorem{lem}[thm]{Lemma}

\newtheorem{defn}[thm]{Definition}
\newtheorem{rem}[thm]{Remark}

\numberwithin{equation}{section}

\def\bC{{\mathbb C}}

\def\bE{{\mathbb E}}

\def\bP{{\mathbb P}}

\def\bR{{\mathbb R}}
\def\bS{{\mathbb S}}

\def\R{{\mathbb R}}

\def\cD{{\mathcal D}}

\def\cH{{\mathcal H}}
\def\cI{{\mathcal I}}

\def\cK{{\mathcal K}}
\def\cL{{\mathcal L}}

\def\cO{{\mathcal O}}

\def\cR{{\mathcal R}}

\def\cU{{\mathcal U}}
\def\cV{{\mathcal V}}
\def\cW{{\mathcal W}}
\def\cX{{\mathcal X}}

\def\chor{{H}}
\def\Ehor{{\cH}}
\def\cver{{V}}
\def\Ever{{\cV}}

\title[]{Curvature Sensitive Cells in the Modular Structure of the Visual Cortex}
\author{Giovanna Citti}
\address{Department of Mathematics, University of Bologna}
\email{giovanna.citti@unibo.it}

\author{Vasiliki Liontou}
\address{Department of Mathematics, University of Bologna}
\email{vasiliki.liontou@unibo.it}

\begin{document}

\begin{abstract}
We propose a model of the functional architecture of curvature-sensitive cells in the primary visual cortex. The model accounts for the modular and hierarchical organization of the cortex, the horizontal connectivity, and the shape of receptive profiles of these cells as Gabor-type filters. We construct a canonical affine subbundle of the cotangent bundle of the manifold of oriented contact elements of the retina as a geometric model for these cells, and show that this subbundle carries an Engel structure related to that of the Cartan prolongation. On an open  submanifold of the Cartan prolongation, we identify generators of the Engel distribution whose iterated Lie brackets span the Lie algebra of $SIM(2)$. 
The identification of $\mathfrak{sim(2)}$ as the Lie algebra of these generators determines $SIM(2)$ as the natural symmetry group for curvature-sensitive cells. 
Finally, we characterize the receptive profiles of curvature-sensitive cells as minima of a  $SIM(2)-$adapted uncertainty principle applied to the generators of the Engel structure.
\end{abstract}
\maketitle

 \tableofcontents

\bigskip

\section{Introduction}

The primary visual cortex (V1) is the first stage of cortical visual processing. Visual stimuli arrive at the cortex as electrical signals propagated from the retinal surface via the optic nerve and lateral geniculate nucleus. As established by Hubel and Wiesel \cite{Hubel, HubelEye}, V1 has a modular and hierarchical organization. Modularity means that the signal is processed by different families of cells, each one sensitive to a different feature— orientation, scale, curvature, motion, and others. The hierarchical organization reflects the fact that visual processing proceeds from simple features, such as edges and contours, toward complex ones, such as curvature and motion.

Mathematically, a cell is a linear filter, characterized by a receptive profile $\Psi:B\rightarrow \bR$, a function on the previous hierarchical layer, whose shape is strictly related to the feature it is sensitive to. The receptive profiles of cells in the same  family form a parametric family  $\{\Psi_m\}_{m \in M}$, and the output of a cell is modeled as a convolution with the visual stimulus $I:B\rightarrow \bR$
$$\mathcal{O}(m) = \int_B \Psi_m \cdot I.$$
This represents the impulse response of the cell to the stimulus. The shape of the receptive profile $\Psi_m$   is constrained by an uncertainty principle reflecting the joint localization of the cell in feature space and in the stimulus domain \cite{DaugmanUncertainty, Marcelja, DaugmanFourier}.  

The interaction between cells and cortico-cortical propagation of neural activity  of the same module takes place along  horizontal connectivity.
 
 In the foundational works of Parent and Zucker  \cite{Parent}, Koenderink–van Doorn \cite{VanDoorn}, Hoffman \cite{Hoffman}, and Petitot–Tondut \cite{Tondut}, horizontal connectivity is modeled by a fiber bundle over the retinal plane equipped with a non-integrable hyperplane distribution —specifically a contact structure. In this framework, an image contour lifts canonically to the cortical module as an integral curve of the contact distribution, representing the pattern of neural activation induced by the stimulus. This geometric description accounts for contour integration and the long-range propagation of cortical activity, and has motivated the use of contact geometry, sub-Riemannian geometry, and Cartan's method of moving frames in vision modeling.

In 2006, Citti and Sarti \cite{FunctArch} introduced a Lie-theoretic approach to the functional architecture of V1, representing the orientation module in terms of stratified Lie groups equipped with a sub-Riemannian metric, chosen to reflect the local symmetries of the cortex. In particular, they identified the orientation module with the group $SE(2)$ of oriented isometries of $\bR^2$, which carries a natural contact structure. 
Sub-Riemannian geometry and properties of geodesics in this setting were used to identify good models of connectivity and association fields \cite{Cuspless, Cuspless2}. As a Lie group it enabled the application of harmonic analysis — notably Generalized Wavelet Transforms — to vision models \cite{Lee, DuitsSE, SETransform}. The identification of $SE(2)$ as the correct symmetry group also led to the characterization of orientation-selective receptive profiles via an $SE(2)$-adapted uncertainty principle \cite{Uncertainty}. More recently, the principle that visual processing should respect the symmetries of $SE(2)$, or of more general groups, has become a cornerstone of geometric deep learning \cite{Cohen, Acosta}.  The assumption that the retina is flat was  removed in \cite{Mashtakov} where the retina was modelled as a half sphere, and then in \cite{MarcLio} where a  more general model without a Lie group structure was presented. In this case the  appropriate space for parametrizing orientation detecting simple cells is the manifold of unit covectors $M=\bS(T^*B)$ over the retinal surface $B$ ({ analogous to $SE(2)$ for $B=\bR^2$}). 

 Combined geometric-signal analytic models appear successful in stimulus analysis and reconstruction via simple features such as orientation and scale because they allow good encoding and decoding of the visual stimulus \cite{ DuitsSE,Scale, Duits, SETransform}. 
 
  Curvature is a complex feature, which is processed by cortical cells pooling
inputs from orientation detection cells in the primary visual cortex. Several authors have approached aspects of this problem from different directions. Zucker and collaborators \cite{  Parent, Zucker, Jonas, ZuckerTexture},  use Cartan's moving frames to study curvature on curves as well as textures, laying the foundation for a computational treatment of curvature.
Petitot in \cite{PetitotCurv} has used 2-jet spaces of $\bR^2$ and their Engel structure to parametrized curvature in a coordinate dependent model while Citti and Sarti \cite{CortArch} extended  this model  by indicating the Engel group as the appropriate Lie group structure.  Sharma and Duits in \cite{Duits}  construct a $SIM(2)$ wavelet transform and show that the associated diffusion  improves curvature estimation compared to the analogous 
technique applied on the SE(2) transform, since it includes multiple scales. In a differential geometric treatment \cite{Bellini}, the space of curvature radii over a Riemannian surface $B$ is identified with its tangent bundle without the 0-section, carrying an Engel distribution. The space is identified with $SIM(2)$ when the surface is the Euclidean plane. 

While a first relation between the fiber bundle and the  sub-Riemannian approach was provided by 
 \cite{Cuspless, Cuspless2} in terms of cuspless geodesics as models for front propagation,
 the  precise connection between models, based on fiber bundle structures,  sub-Riemannian techniques, or wavelet transform methods,   remains to be clarified.

The scope of this work is to construct a modular model of the visual cortex, where  curvature detecting cells are described through the same mathematical mechanisms used in models of the previous processing layer- orientation selective cells. In analogy to the mechanism of orientation detection, a neuron should be represented by a mathematical object that ``extracts'' the desired feature. 
This leads to a model of iterative fiber bundles where each layer is a bundle of covectors of the previous layer.
Each layer is equipped with a non-integrable plane distribution modeling the horizontal connectivity within the module. Associated to this iterative bundle structure there is an iteration of integral transforms modeling the functionality of the cells as filters.  The association is established through the Lie algebra closure of the generators of the non-integrable distribution. 
The architecture of the cortex, in this reading, has a single repeated pattern.

 Our first objective is to identify the relation between fiber bundle and sub-Riemannian approaches for curvature detecting cells. The manifold $M$ which describes simple cells detecting orientation is a Legandrian circle bundle and the key observation is that it carries two transversal hyperplane distributions: one is the contact distribution $H$ and the other is the Ehresmann distribution induced by the metric connection. We use this observation to establish an intrinsic description of the directions of propagation of the neural activity within $M$- between neurons of similar orientation sensitivity. We use the standard decomposition of the velocity vector field of the Legendrian lift of a retinal contour (B-valued regular curve) into a vertical and a horizontal component, to ``locate'' the curvature on the vertical component of the contact distribution. Then, we identify the space of curvature detecting cells as a special affine subbundle of the dual $H^*$ of the contact distribution.  Finally, we show that the appropriate choice of covector on $H^*$ corresponding to a point in $M$ is characterized by a plane  distribution obtained as the kernel of a 1-form on $T^*M$. This mirrors, one level up, the procedure by which the contact structure on $M$ arises.
 
 The second objective is to build a bridge between the geometric modeling and signal analysis performed by the cells by determining the right Lie group structure.  We associate the previous construction  to the Cartan prolongation of $M$ which naturally carries an Engel structure. We show that the distribution on the special affine bundle of curvature detecting cells is Engel. We identify a pair of vector fields generating this distribution whose Lie algebra closure is isomorphic to $\mathfrak{sim}(2)$, the Lie algebra of the group of planar similitudes $SIM(2)$. 
In this way we recover for the description of curvature the this same group which emerged independently in models  \cite{Duits, Bellini}, 
based on different assumptions.

The third objective is to introduce signal analysis methods that are compatible with the geometric setting. We construct a family of curvature-sensitive receptive profiles extending the Gabor-type filters used for orientation detection to a two-layer architecture.    As curvature-selective cells receive in input the output  of the  $SE(2)$ transform (whose range will be denoted $\cR$), 
we use the quasi-regular representation of $SIM(2)$ on square integrable functions of $SE(2)$  to generate the family of curvature-selective receptive profiles.
Under a $SIM(2)-$covariance condition on the receptive profiles, they define a ${SIM}(2)$-wavelet transform with an effective kernel $f\in L^2(\bR^2)$.  This effective kernel represents the detectable receptive profile on the 2-dimensional projection of the retina in the cortex, although processing occurs in a higher cortical level.  Finally, we characterize the shape of the effective profiles using an uncertainty principle for $SIM(2)$ and, in particular, for the operators induced by the left-invariant generators of the Engel structure.

\section{Preliminaries: Orientation Selective Odd Simple Cells  }\label{Preliminaries}
In this chapter, we review models of Odd Simple Cells (OSCs) selective to orientation. The purpose of this review is to recall these models which will serve as the foundation for constructing a geometric framework for curvature-selective cells and use a uniform language for the two. We focus in particular on totally intrinsic model, independent of any reference axis or a coordinate system  apriori fixed, since there is no such a reference in the brain.

We consider the retina to be modeled by an oriented, 2-dimensional, Riemannian surface $(B, g)$ with metric $g$. 
  The functional architecture of  Orientation-Position OSCs is modeled by the oriented projectivized cotangent bundle of the retina--that is the set $\bS(T^*B)$ of non-
zero covectors in $T^*B$ where two covectors are identified if they differ by a positive real number,
 $$\bS(T^*B) = (T ^*B \backslash \{0\})/\bR^+$$
where $\{0\}$ is the zero section of $T^*B$ and $\bR^+$ denotes the positive real numbers. The existence of the metric $g$ allows us to identify the projectivized cotangent space of the retina with the unit covector bundle of the retina and so we will be using them interchangeably throughout. 
We will denote this space of Orientation-Position Odd Simple Cells by $M=\bS(T^*B)$. Each point \( q \in M \) is a pair \( (b, p) \), where \( b \in B \) is a base point and \( p \in T_b^*B \) is a unit covector. We will also use $\pi$ to denote the natural projection
\[ \pi: M \to B , \; \;  \pi(b, p) = b .\] The total space \( M \) forms an \( \mathbb{S}^1 \)-bundle over \( B \), as each fiber \( \pi^{-1}(b) \) is the space of unit covectors at \( b \), which is diffeomorphic to \( \mathbb{S}^1 \).

\paragraph{\textbf{Gauge Coordinate System}}
  As an oriented $\bS^1-$bundle, $M$ is a principle $\bS^1-$bundle. Precisely, the metric $g$ on $B$ induces a metric on each fiber $\pi^{-1}(b)$ of the unit cotangent bundle $\bS(T^*B)$ and therefore a length for a curve along the fiber. At each $b\in B$, we denote the length of the fiber $\pi^{-1}(b)$ (as a curve) to be $\ell_b$. Using the associated distance in the fiber,  the right action of $\bS^1$ on $M$ is 
    \begin{equation}
        \bS^1\times M \ni (\theta, q)\mapsto \theta\cdot q
    \end{equation}
    where $\theta\cdot q$ is the point on the fiber $\pi^{-1}(b)$ of $q$ whose distance along the fiber with respect to the orientation is $\frac{\theta}{2\pi }\ell_b$.

We consider gauge coordinates on $M$ by fixing a nowhere-vanishing covector field on a local chart of \( (U, (x,y)) \) of $B$, for instance \( dx \). Any covector \( q \in\pi^{-1}(U)\subset M \) can then be obtained by the action of $\bS^1$ on \( dx \). This defines a coordinate chart \( (\pi^{-1}(U) \setminus \{dx\}, \theta_x) \), where \( \theta_x(q) \in \mathbb{S}^1 \) is the unique angle such that
\begin{equation}\label{coorch}
\theta_x(q) \cdot dx = q.
\end{equation}
We denote by \( J \) the almost complex structure \( J: T_b B \to T_b B , J^2=-I\), compatible with $g$ and the orientation 
  By applying the dual operator $J^*: T^*B\rightarrow T^*B,~ J^*(q)(-) := q(J(-)) $ on covectors, we obtain another reference covector field \( J^* dx \). This yields a second chart \( (\pi^{-1}(U) \setminus \{J^* dx\}, \theta_J) \), where for \( q \in M \), the angle \( \theta_J(q) \in \mathbb{S}^1 \) satisfies
\[
\theta_J(q) \cdot J^* dx = q.
\]
Working similarly on can form an atlas that covers \( M \). When the reference covector field is understood from context, instead of $(x,y,\theta_x)$ we will use the notation  \( (x, y, \theta) \) for gauge coordinates on \( M \).

\subsection{Contact Structure and Lifting Contours from the retina to the Cortex}
To model the horizontal connectivity between OP-OSCs, the authors of \cite{FunctArch} assumed the retina to be identified with $\bR^2$ and considered a contact distribution $\chor$ on \( \bS(T^*\bR^2)\), defined in gauge coordinates $(x,y,\theta)$ as the kernel of the 1-form 
\begin{equation}\label{ContactForm}
    a = -\sin(\theta)\,dx + \cos(\theta)\,dy.
\end{equation}

 This contact structure allows for horizontal lifting of smooth retinal contours, namely curves \( \gamma: I\subset\bR \rightarrow \mathbb{R}^2 \) to cortical curves \( \Gamma: I \rightarrow \bS(T^*\bR^2) \) such that \( \dot{\Gamma}(t) \in \chor_{\Gamma(t)} \) for all \( t \in I \).

While there is a canonical contact structure on $\bS(T^*\bR^2)$, obtained by restricting the canonical Liouville form $\lambda\in T^*(T^*\bR^2), ~ \lambda_q=(d\pi)^* (q)$ to the unit covectors, it does not coincide with the contact distribution $\chor$ of the model. A horizontal lifting of curves with respect to the canonical contact distribution on $\bS(T^*\bR^2)$ lifts to covectors that annihilate the tangent vector of the original curve, while a lift with respect to $H$ lifts a curve to covectors that align with the tangent vector of the curve.

A coordinate free definition of the  contact distribution $\chor$ for an oriented Riemannian base $B$ was given in terms of a \textit{twisted Liouville form}  in  \cite{MarcLio} and we adopt the same definition here. Precisely, we consider the \textit{twisted Liouville form}
$$\lambda_J\in T^*(T^*B), ~
(\lambda_J)_q = (d\pi)^*(-J^*q),
$$
and take its restriction to  $M$, which is the contact 1-form \begin{equation}
    a_J\in T^*M, ~a_J:=(\lambda_J)|_{M}.
\end{equation}
Then the contact distribution $H$ is given by 
\begin{equation}\label{ContactDistribution}
    \chor = \ker(a_J).
\end{equation}
  The quotient bundle $TM/ H$ is trivial since $H$ is coorientable. It can be identified with the orthogonal complement $V:=H^{\bot_{g_M}}$  of $H$ 
with respect to the metric $g_M$ on $M$ induced by the metric $g$ on $B$. We will refer to $H$ as the horizontal distribution with respect to the contact structure, or contact-horizontal distribution while the distribution 
$V$ will be respectively called the vertical distribution with respect to the contact structure, so that we can write 
\begin{equation}
    TM\simeq \chor \oplus V
\end{equation}

This definition allows for an \emph{angle-based} horizontal lift of curves, as intended by the coordinate based model \cite{FunctArch}. 
  Indeed,  we consider the lifted curve $\Gamma$ of  a retinal curve $\gamma: I\rightarrow B$ as 
\begin{equation}\label{HorLift}
    \Gamma(t)=\{q=(b,p)\in M: b=\gamma(t) \text{ and } \ker(J^*p)=T_b\gamma\text{ and }p(\dot\gamma)>0\},
\end{equation}
that is, for each time \( t \), the lifted covector \( p \) is the unique element in the fiber over \( \gamma(t) \) whose  twisted covector $J^*p$  annihilates $T_b\gamma$ and is positively oriented. The lifted curve $\Gamma$ is horizontal with respect to $\chor$. \footnote{This assertion simply means  that \( a_J(\dot{\Gamma}(t)) = 0 \), which follows directly from the construction.}

This lift is aligned with the tangent angle of the velocity vector field of the curve $\gamma$. We will use the standard notation for the flat operator $\flat : TB\rightarrow T^*B$ with respect to $g$ and the sharp operator $\sharp: T^*B\rightarrow TB$.
\begin{lem}
    Let $\gamma:I\rightarrow B$ be a regular retinal curve, $u:I\rightarrow \bS B, ~ t\mapsto \frac{\dot{\gamma}(t)}{\|\dot{\gamma}(t)\|}$ its unit tangent vector field and $\Gamma: I\rightarrow M$ its horizontal lift, then $\Gamma(t)=\flat(u(t)), \text{for every } t$. Moreover, for any $q\in M_{|_\gamma}$ we have 
    $$atan2(g^*(\Gamma, J^*q), g^*(\Gamma, q))= atan2(g(u, J^*\sharp q),g(u, \sharp q))$$
\end{lem}
\begin{proof}
Since $J^*\flat X=-\flat(JX)$ and by the definition of $\Gamma$ we have that \begin{align*}
    0= g^*(J^*\Gamma, u^\flat)=&g^*(\Gamma, -J^*u^\flat)   \\
    =&g(\Gamma^\sharp, -(J^*u^\flat)^\sharp)= g(\Gamma^\sharp, Ju).
\end{align*}
Thus, we obtain that $\sharp\Gamma=au,$ for some $a\in C^\infty(\bR^2)$ positive and therefore $\Gamma=\flat(u)$. Moreover, let $q\in M_|{_\gamma}$, then 
$$atan2(g^*(\Gamma, J^*q),g^*(\Gamma, q))=atan2(g(u,\sharp(J^*q)),g(u,\sharp q))=atan2(g(u, J^*\sharp q), g(u,\sharp q)).$$
\end{proof}

\begin{rem}
While there is no mathematical novelty in the previous lemma, we use it to clarify that there are two equivalent coordinate independent ways to consider the lift $\Gamma$ and which satisfy the purpose of the lift considered in \cite{FunctArch}, namely to lift a retinal contour with respect to the angular coordinate $\theta$. The first way is as in \eqref{HorLift} and the second is using the musical isomorphism.

\end{rem}

\subsection{Receptive Profiles of Orientation Sensitive Simple Cells and the SE(2) transform
}    
In this section, we review the mathematical background of  models for orientation-selective simple cells, in which the response  is  expressed via the SE(2)-transform, a standard construction in harmonic analysis \cite{DuitsSE,CR, SETransform}. The background provided here will only be used later in Chapter \ref{ReceptiveProfs}.

  The group $SE(2)$ of orientation preserving isometries of the plane $\bR^2$ consists of planar translations $t_{(x,y)}: \bR^2\rightarrow \bR^2, ~ (\xi,\eta)\mapsto (\xi,\eta)+(x,y)$, planar rotations $r_\theta:\bR^2\rightarrow \bR^2, ~ (\xi,\eta)\mapsto r_\theta (\xi,\eta), r_\theta\in  SO(2)$. It is an outer semidirect product $SE(2)=\bR^2 \rtimes \bS^1$ with elements $(x,y ,\theta) \in \bR^2 \times \bS^1$ and
composition law
$$(x,y ,\theta) \cdot (x^\prime, y^\prime, \theta^\prime) = ((x,y) + r_\theta (x^\prime,y^\prime), \theta + \theta^\prime).$$

Its Haar measure, that is the (Radon) measure on the group that is invariant under group
operations, is the Lebesgue measure on $\bR^2 \times \bS^1$.
A standard wavelet analysis of two-dimensional signals uses the quasi-regular representation of $SE(2)$ on $L^2(\bR^2)$
  \begin{align}\label{SERep}
    \rho(x,y,\theta)f(\xi,\eta)=&T_{(x,y)}R_\theta f(\xi,\eta)\\=&f(r_{-\theta}(\xi,\eta)-(x,y)),~f\in L^2(\bR^2)\notag
\end{align}
where 
$T_(x,y)$
 and 
$R_\theta$ denote translation and rotation operators. Then the $SE(2)-$ wavelet transform is defined with respect to this representation as follows.
\begin{defn}
    Let $\Psi_0\in L^2(\bR^2)$, the $SE(2)-$wavelet transform on $L^2(\bR^2)$ is the operator
    \begin{equation}\label{SETransform}
    \cO_{\Psi_0}(x,y,\theta)I=\int_{\bR^2}I(\xi,\eta)\overline{\rho(x,y,\theta)\Psi_{0}(\xi,\eta)}d\xi d\eta =\langle I, \rho(x,y,\theta)\Psi_0\rangle_{L^2(\bR^2)}.
\end{equation}

\end{defn}
The function $\Psi_0$ is often referred to as the \textit{mother window}(see, for example, \cite{Wavelets}, \cite{Wavelet2}, \cite{Wavelet3}, \cite{Fuhr}, \cite{DuitsSE}  and references therein).

 Orientation–position receptive profiles correspond to the analyzing wavelets of the 
$SE(2)$ wavelet transform \cite{DaugmanUncertainty, DaugmanFourier,Lee, FunctArch, SETransform}. The receptive profile selective to position $(x,y)$ and orientation $\theta$ is obtained from a mother window $\Psi_0$ through the group representation ${\rho(x,y,\theta)\Psi_0}$. Thus, the family of all receptive profiles is 
\begin{equation}\label{OrientationReceptiveProfiles}
    {RP_{SE(2)}}:=\{\Psi_{(x,y,\theta)}=T_{(x,y)}R_\theta \Psi_0: (x,y,\theta)\in \bR^2\times\bS^1\}.
\end{equation}
 For a visual stimulus $I\in L^2(\bR^2)$, the response of a receptive profile $\Psi_{(x,y,\theta)}$ is  the value of the $SE(2)-$transform $\cO_{\Psi_0}(x,y,\theta)I$.

 The transform $\cO$ is a bounded operator from $L^2(\bR^2)$ to $L^2(SE(2))$ if $\Psi_0 \in L^1(\bR^2)\cap L^2(\bR^2)$ but it cannot be injective \cite{Wavelets}. 
 
 For image analysis and reconstruction, invertibility is required. It is
possible to obtain a bounded and injective transform by reducing the
wavelet analysis to the space of bandlimited functions, namely functions whose Fourier
transform is supported on a bounded set,
\begin{equation*}
    PW_R:=\{f\in L^2(\bR^2): supp\hat{f}\in B_R\}, \text{ where }B_R = \{\xi \in \bR^2 : |\xi| < R\}. 
\end{equation*}
In this case, the image
of the $SE(2)$ transform is a reproducing kernel Hilbert subspace of $L^2(\bR^2\times \bS^1)$ (see \cite{SETransform}) which we will denote as $\cR \subset L^2(\bR^2\times \bS^1)$ and the $SE(2)-$transform \ref{SETransform} is an injective and bounded operator on $PW_R$ if and only if there exist constants $A,B>0$ such that 
\begin{equation}\label{Calderon}
    A<\int_{\bS^1} |\hat{\Psi}_0(R_\theta^{-1}\xi)|d\theta <B.
\end{equation}
For what follows we will consider the domain of $\mathcal{O}_\Psi(x,y,\theta)$ to be $PW_R$.

The cortical activity $I\in L^2(\bR^2)$ defined on the 2D cortical layer is propagated anisotropically along the horizontal long
range connectivity, which is effectively described in terms of the Lie algebra of $SE(2)$.  
The Lie algebra of $SE(2)$ has dimension three and it is generated by the infinitesimal transformations
along two orthogonal translations and the rotation of the group.  In order to take into account
the anisotropy of the connectivity, in  \cite{Tondut} and \cite{FunctArch} only two generators were considered instead of three. In
particular in \cite{FunctArch} it was proposed to choose as generators the left invariant vector fields $X_1$ and $X_2$ representing
only one infinitesimal translation and the infinitesimal rotation at a point $(x,y,\theta)$, 
\begin{equation}\label{vectorfieldscoord}
    X_1= cos({\theta})\partial_{{x}}+sin({\theta})\partial_{{y}},~ X_2=\partial_{{\theta}}.
\end{equation}
The differentiated representation $d\rho$ of the quisi-regular representation maps 
 the vector fields $X_1$ and $X_2$ to Skew-Adjoint operators on $L^2(\bR^2)$,
\begin{align*}\label{InfinitesimalAction}
    d\rho: &\mathfrak{se(2)}\rightarrow SkewAdj(L^2(\bR^2)),\\ &X_1\mapsto d\rho(X)f=\frac{d}{dt}|_{t=0}\rho(exp(tX))=\partial_{x_1}f(x_1,x_2)\\
    &X_2\mapsto d\rho(X)f=\frac{d}{dt}|_{t=0}\rho(exp(tX))=(x_1\partial_{x_2}-x_2\partial_{x_1})f(x_1,x_2)~, \text{ for } f\in L^2(\bR^2).
\end{align*}
Moreover, $X_1$ and $X_2$ do not commute and the non-commutativity of operators $d\rho(X_1)$ and $d\rho(X_2)$ is analogous to the non-commutativity of the momentum and position operators in the canonical quantum mechanical case, where the operators are obtained from generators of the polarized Heisenberg algebra with respect to the Schr\"odinger representation.  In that specific situation the variance
of the position and momentum operators on functions have a lower bound that has been formalized in
the well known Heisenberg uncertainty principle. The uncertainty principle has been formally extended (see Theorem 2.4,\cite{FollandUncertainty}) to connected Lie groups acting on Hilbert spaces. This extension was used in \cite{Uncertainty} to determine the shape of orientation-position-receptive profiles under the assumption that they should minimize the uncertainty principle 
\begin{equation}\label{UncertaintyEq}
    \|d\rho(X_1) u\|_{L^2(\bR^2)} \|d\rho(X_2) u\|_{L^2(\bR^2)} \geq \frac{1}{2}
 |\langle d\rho([X_1,X_2]) u, u\rangle _{L^2 (\bR^2)}|.
\end{equation}

Under this assumption, the mother receptive profile is of the form \begin{align}\label{MotherProfile}
     \Psi_0\in&~ L^2(\bR^2), \nonumber\\ ~&\Psi_0(x)=s\sqrt{2}\pi e^{-2\pi i p x_1}e^{-2\pi^2 s^2 |x|^2}, s,p\in \bR^+ .
 \end{align}

\section{The  Ehresmann connection and cells sensitive to Curvature }\label{Curvature}
The purpose of this chapter is to identify the appropriate space that represents the functionality of curvature detecting cells, in the modular organization of the cortex and in a mechanism that is parallel to that of orientation detection and it naturally extends it. We start with a key observation: since the model space of the horizontal connectivity of the simple cells is a Legendrian fiber bundle- its fibers are horizontal with respect to the contact structure- one can use the Ehressman connection on the unit cotangent bundle of the retina to provide  an intrinsic description of the activity propagation within the cortex (Defintion \ref{IntrinsicDirections}). 
Using an intrinsic description of propagation directions, we show that curvature is "located" in the vertical subbundle of the tangent bundle of $M$ (Proposition \ref{decomposition} ). As a result curvature detection should be represented by linear forms acting on this space. Finally, we introduce a special affine bundle over $M$ which is appropriate for modeling curvature detecting cells (Theorem \ref{CurvatureDetection}) . 

\subsection{Ehresmann Connection and Characterization of Retinal Curves lifted to M}

In this section, we adopt the standard differential geometric notation (see \cite{Morita,Tu,LeeCurv} for background).  We will consider the Levi-Civita connection $\nabla:TB\times TB\rightarrow TB$ induced by $g$ and we denote the parallel transport on $TB$ along a curve $\gamma:[0,1] \to B$  by $P_\gamma^{t_0, t}: T_{\gamma(t_0)}B\rightarrow T_{\gamma(t)}B$. Note that since the parallel transport preserves the metric, $P_\gamma^{t_0, t}$ restricts to a parallel transport on the unit tangent bundle $\bS B$, which we will denote as
$$P^{\bS, t_0, t}_\gamma:=P_\gamma^{t_0 t}|_{\bS B}: \bS_{\gamma(t_0)} B\rightarrow \bS_{\gamma(t)}B.$$
We will denote the covariant derivative of a vector field $X: B\rightarrow \bS B$ along a curve $\gamma$ at a point $ \gamma(t_0)$ as 
\begin{equation*}
    D^\bS_{t_0}X=\frac{d}{dt}\Big((P^{\bS,t_0, t}_
    \gamma)^{-1}\circ X_{\gamma(t)}\Big)_{|_{t=t_0}} \in T(\bS_{\gamma(t_0)}B).
\end{equation*}

Finally we recall that the covariant derivative induces an Ehresmann connection $\cH\subset TM$, by defining $\cH_q$, for $q \in M$ with $b=\pi(q)$, to be the image $dX_b(T_b B)$ where $X$ is a section of $M$ with $X(b) = q$ and $D_YX = 0$ for all $Y \in T_b B$. It is a smooth distribution field of subspaces $\cH_q\subset T_qM$ which are horizontal with respect to the fiber bundle $M\rightarrow B$, that is  $$\cH_q\oplus \cV_q= T_qM, \text{ for every } q\in M.$$

To distinguish it from the contact distribution $\chor$ introduced in \ref{ContactDistribution} we will be referring to the distribution $\cH$ as Ehresmann-horizontal.  

\subsubsection{Intrinsic description of Propagation of Neural Activity}
    One should note that we introduced contact-horizontal and vertical distributions $\chor$ and $\cver$, as well as  Ehresmann-horizontal and Ehresmann-horizontal sub-bundle $\cH$ and $\cV$, providing $TM$ with two distinct splittings,
    \begin{equation*}
        TM=V\oplus \chor \text{ and }TM=\cH\oplus \cV.
    \end{equation*} 
We will now use these two splittings to provide an intrinsic description of the directions of propagation of the neural activity, which are described in coordinates in \cite{FunctArch}. To provide a formal characterization we need the following two lemmas.

\begin{lem}
    The intersection of the horizonal distributions  
    $\cL=\chor\cap\cH$ is a line subbundle of $TM$. Its fiber defines a privileged line $\cL_q= H_q\cap\cH_q$ at each tangent space $T_qM$.
\end{lem}
\begin{proof}
    The vertical subbundle $\cV$ is contained in the contact distribution $H$ and since both subbundles $\chor$ and $\cH$ have codimension 1,  they are transverse in $TM$. Consequently, their fiberwise inersection $\chor\cap\cH$ is a subbundle of $TM$ whose fiber is  the line $\cL_q=\chor_q\cap \cH_q$.
\end{proof}

In the next lemma we see that the tangent bundle $TM$ can be expressed as a direct sum of three intrinsic sub-bundles, with the same base as $M$.
\begin{lem} The contact distribution $\chor$ splits to  $$\chor = (H\cap\cH) \oplus \cV.$$
     Moreover, the tangent bundle $TM$ is isomorphic to the Whitney sum
    $$TM \simeq (H\cap\cH) \oplus \cV \oplus V$$

\end{lem}
\begin{proof}
    The contact distribution $\chor$ fits into the short exact sequence of vector bundles
    \begin{equation*}
        0\rightarrow\chor\cap \cH\rightarrow \chor\rightarrow \chor/ \chor \cap \cH\rightarrow 0.
    \end{equation*}
 Now, since the vertical subbundle $\cV\subset TM$ is trivial and $\cV\cap(\chor\cap \cH)=\{0\}$ by definition, it follows that any global section of $\cV$ induces a global section on $\chor/\chor\cap \cH$ and therefore we can identify $\cV$ with the quotient bundle  $\chor/\chor\cap \cH$ and obtain the short exact sequence
 \begin{equation*}
     0\rightarrow\chor\cap \cH\rightarrow \chor\rightarrow \cV \rightarrow 0.
 \end{equation*}
 Finally, this short exact sequence splits since $\cV$ is a subbundle of $\chor$, via the bundle map $\cV\rightarrow \chor$ and therefore $\chor=\cV\oplus (\chor\cap \cH) $.
\end{proof}

As a result,  this splitting allows us to introduce an intrinsic orthogonal basis of the contact-horizontal plane. Using the induced metric $g_M$ on $M$, we can call a unitary vector any vector $X$ such that and $g_M(X,X)=1$, and introduce an orthonormal basis as follows.
\begin{defn}\label{IntrinsicDirections}
    We introduce a basis 
$$(X_1, X_2)$$
of the contact-horizontal distribution $\chor$ in the following way. 
 $X_1$ is the unique unitary vector field such that $X_1(q)\in \cH_q \cap \chor$ while $X_2$ is the unique unitary vector field such that $X_2(q)\in {\cV}_q $.
\end{defn}
We used the same notation as in 
\eqref{vectorfieldscoord} since it is easy to show that this is a coordinate free representation of the same vector fields. The key property of this splitting is that it is unique and only dependent on the metric connection and the contact structure. Consequently, it induces a unique natural choice of coordinates, and proofs in this choice of coordinates provide intrinsically determined objects. 

\subsubsection{Singed Curvature in Horizontal Lifts of Retinal Contours}

Recall that the horizontal lift of a retinal curve $\gamma: I\rightarrow B$ with respect to the contact distribution $H$ is a section $\Gamma: \gamma\rightarrow M$ along the path $\gamma$.
Recall also that the velocity vector $\dot{\Gamma}$ of the lifted curve $\Gamma:I\to M$  at each $t\in I\subset \bR$
is  the pushforward 
$\dot \Gamma (t) = \Gamma_{*} \left( \frac{d}{dt} \right)  $ acting on functions in $f\in C^\infty(M,\bR)$ by $$\dot{\Gamma}(t)f=\frac{d}{dt}(f\circ\Gamma)(t), \forall f\in C^\infty(M,\bR)$$ and therefore it is a vector in $T_{\Gamma(t)}M$. We will now relate $\dot{\Gamma}$
 to the signed curvature of the original curve. As in the previous section, we consider $\gamma(0)=b$ and $u(t)=\frac{\dot{\gamma}(t)}{\|\dot{\gamma}(t)\|}$. 

Let   $N(t)$ be the unique normal vector field such that $(u(t), N(t))$ is an oriented orthonormal basis of $T_{\gamma(t)}\bR^2$, we recall that the \textit{ signed curvature } $\kappa_N$ of $\gamma$ at $\gamma(t)$ is given by 
$$\kappa_N(t)=\langle D_tu(t), N(t)\rangle.$$

By differentiating $\|u(t)\|^2=1$ we obtain a relation between the signed curvature of  the curve $\gamma$ 
and the 
the covariant derivative of $u$: 
\begin{equation}\label{covcurv}
   D_tu=\kappa_N(t)N(t).
\end{equation} 

We recall the differential-geometric notions needed for the following proposition.
The covariant derivative along $\gamma$  maps unit vectors of $B$ to vertical vectors of $T(\bS B)$, namely $D_{t_0}^\bS: \bS_{\gamma(t)}B\rightarrow \cV_{{u}(t)}\bS B$.
The  connection map (vertical projection) $K:TM\rightarrow \cV$ sends each $T_qM$ to the vertical subspace $\cV_q$ by projecting along $\cH_q$.  The vector space structure on each fiber $T_b^*B$ yields a natural
isomorphisms $$Vert_q:T_b^*B\rightarrow \cV_q, ~a\mapsto \frac{d}{dt}(a+ tq)_|{_{t=0}}, ~ q\in T_b^* B.$$
Finally, the differential $d\pi_q:T_qM\rightarrow T_bB$ of the projection $M\rightarrow B$  restricts to a linear isomorphism $(d\pi_q)_{|_H}:\cH_q\rightarrow T_b B$, whose inverse is the \textit{horizontal lift} $Hor_q: T_b B\rightarrow H_q$.

\begin{prop}\label{decomposition}
    Let $\gamma:I \to B$ be a retinal curve in $B$ and $\Gamma: I\rightarrow M$ its horizontal lift with respect to the contact distribution $H=ker(a_J)$. Then
    \begin{equation}
        \dot{\Gamma}= |\dot\gamma| X_1+ \kappa_N X_2
    \end{equation}
    where $\kappa_N$ is the curvature of $\gamma$.
\end{prop}

\begin{proof}
The contact horizontal lift $\Gamma$ of a curve $\gamma$ is a section $\Gamma:\gamma\rightarrow M$ along $\gamma$ which for every $t\in I$ is given by $ \gamma(t)\mapsto \flat(u(t))$. Moreover, its total differential $\dot{\Gamma}$ is a section $\dot{\Gamma}: \Gamma\rightarrow TM$ of $TM\rightarrow M$ along $\Gamma$ which for every $t$ is given by $\dot{\Gamma}(t)=\Gamma_*(\partial_t)$ and can be written with respect to the vertical projection $K$ and the Horizontal lift $Hor_\gamma$ as
$$\dot{\Gamma}(t)= Hor_{\Gamma(t)}(\gamma_*(\partial_t))+K(\Gamma_*(\partial_t)).$$
But for the section $\Gamma$ along the path $\gamma$, the vertical component of its velocity vector field $\dot{\Gamma}$ is equal to the covariant derivative of $\Gamma$ along $\gamma(t)$, namely
$$K(\Gamma_*(\partial_t))=D^\bS_t(\Gamma)=D^\bS_t(\flat u) \in V_{\Gamma(t)}.$$

We apply the linear isomorphism $Vert_\Gamma^{-1}:V_\Gamma\rightarrow T_\gamma^*B$ to identify $D^\bS_t(u)$ to the covariant derivative of $\flat u$ along $\gamma$, as a section of the vector bundle $T^*B$. 
Moreover, since the musical isomorphism and the covariant derivative commute, we obtain
$$K(\Gamma_*(\partial_t))=\flat(D^\bS_t(u)),$$
and by the definition of signed curvature the vertical component becomes $$K(\Gamma_*(\partial_t))=\flat(D_tu)=\flat(\kappa_n(t)N(t)).$$
Finally, by the definition of the metric $g_M$ on $M$ we have that \begin{align*}g_M(K(\Gamma_*(\partial_t)),K(\Gamma_*(\partial_t)))&= g^*(\flat (\kappa_N(t)N(t)), \flat (\kappa_N(t)N(t)))\\&=\kappa_N(t)g(N(t),N(t))=\kappa_N^2(t)\end{align*} and since $X_2$ is the unique vertical unitary vector in $V_{\Gamma(t)}$ we conclude that  $$K(\Gamma_*(\partial_t))=\kappa_N(t)X_2.$$ 

Similarly, for the horizontal component we know that $Hor_{\Gamma(t)}(\gamma_*(\partial_t))\in \cH_{\Gamma(t)}\cap H_{\Gamma(t)}$ and $g_M(Hor_{\Gamma(t)}(\gamma_*(\partial_t)), Hor_{\Gamma(t)}(\gamma_*(\partial_t)))=g(\gamma_*(\partial_t), \gamma_*(\partial_t))=\|\dot{\gamma}\|$ and since $X_1$ is the unique unitary vector in $\cH_{\Gamma(t)}\cap H_{\Gamma(t)}$ for every $t$,
$$Hor_{\Gamma(t)}(\gamma_*(\partial_t))=\|\dot{\gamma}(t)\|X_1.$$
\end{proof}

The relation between the contact and connection horizontal distributions established above is important for applications in visual perception, as it further unifies two modeling frameworks in the literature: the fiber bundle structure of the cortex endowed with a connection represents the hierarchical structure of the cortex, and the contact structure represents the long-range connectivity. 

It also yields a strong characterization of contact horizontal lifted curves, which are the only curves compatible with cortical connectivity. In particular, considering $B=\bR^2$, as in \cite{Cuspless,Cuspless2}, the  integral curves of vector fields with a non-trivial component along $X_1$ are regular and have regular projections onto $\mathbb{R}^2$. This provides an intrinsic characterization of \emph{cuspless geodesics} used in these works, namely geodesics in $SE(2)$ whose projections onto $\mathbb{R}^2$ have no cusps. Such curves were used in  to model association fields underlying perceptual organization in psychophysics.

\subsection{Curvature Selective Cells: Lifting from M to an affine sub-bundle of T*M}\label{CurvatureSpace}

 In previous sections, 
 we 
 established that the signed curvature $\kappa_N$ of a
 retinal curve $\gamma$ can be located in the vertical subbundle of $TM$. We will assume that curvature selective cells work in a similar manner to the orientation selective cells. Thus, while the feature of interest is located in the tangent bundle the  cells in order to extract said feature  should be located
 in the dual bundle. To capture this mechanism,  we will now show how to lift  $\Gamma$ to a curve $G$ on the appropriate sub-bundle of $T^*M$. Through this lift we obtain the neurons reacting to the given stimulus. The argument will be exactly analogous to the  lift, from $B$ to $M=\bS(T^*B)$ of $T^*B$, used to model orientation detection, in order to strongly underline the modularity of the cortical areas.

We start by imposing the differential constraint coming from the Legendrian lift from $B$ to $M$, namely we will only consider tangent vectors in  the contact-horizontal distribution $\chor=ker(a_J)$, and we split accordingly tangent and cotangent space. We recall, that the  connection on $(M,g_M)$ induces a split of the short exact sequence 
\begin{equation*}
    0\rightarrow\pi^*(T^*\bR^2)\rightarrow T^*M\rightarrow V^*\rightarrow 0,
\end{equation*}
 thus $V^*$ is a subbundle of $T^*M$ identified to the annihilator of $H\subset TM$, namely
 \begin{equation*}
   \cver^* \simeq  null(\chor):=\{Q\in \Omega(M): Q(X)=0 \text{ for every }X\in \chor\}
. \end{equation*}
Similarly, we consider $H^*\simeq null(V)$ such that 
$$T^*M = \cver^*  \oplus \chor^* .$$

We will denote by  $\Pi:T^*(T^*M)\rightarrow T^*M$ the projection onto $T^*M$ and  $\lambda_{T^*M}$ the tautological 1-form on $T^*M$,  namely $$({\lambda_{T^*M}})_{Q}(W)= (d\Pi)^*(Q)(W), \;\;W\in T_Q(T^*M).$$ 
 We consider its restriction to $H^*$
 \begin{equation}   a_{H^*}:=(\lambda_{T^*M})|_{H^* }
\end{equation}
which induces the singular distribution of subspaces of $T_qH^*$ given by
\begin{equation}
    ker({a_{H^*}}_q).
\end{equation}
Next, we will determine a submanifold of $H^*$ where the restriction of $a_{H^*}$ is non-vanishing and introduce a horizontal lift of a curve $\Gamma\subset M$ with respect to the distribution induced by the restriction of $a_{H^*}$.

We already noted in the previous section that it is necessary to use together the contact-horizontal distribution $\chor$ and Ehressmann-horizontal one $\Ehor$ to obtain a coordinate independent and intrinsic characterization of the  directions in the tangent bundle. In perfect analogy we will consider  the dual distribution in the dual space. 

The connection on $M$ also induces a split of the short exact sequence 
\begin{equation*}
    0\rightarrow\pi^*(T^*\bR^2)\rightarrow T^*M\rightarrow \Ever^*\rightarrow 0,
\end{equation*}
 thus $\Ever^*$ is a subbundle of $T^*M$ identified to the annihilator of $\Ehor\subset TM$, namely
 \begin{equation*}
   \Ever^*\simeq  null(\Ehor):=\{Q\in \Omega(M): Q(X)=0 \text{ for every }X\in \Ehor\}
. \end{equation*}
Since at every $q\in M$, $V_q$ is the orthogonal complement of $H_q$, and $\cH_q$ is the orthogonal complement of $\cV_q\subset H_q$, then $V_q\subset\cH_q$ and $null(\cH_q)\subset null(V_q)$
so that we have, 
$$\chor^*  = \Ever^* \oplus   ({\Ehor}^* \cap {\chor}^*).$$

Let us consider a curve $\Gamma: I\rightarrow M,$ which is the Legendrian lift of a curve 
$\gamma: I\rightarrow B$
and for every $t$ let us consider the set 
\begin{equation}
S= \{(q,P)\in \chor^*: q=\Gamma(t) \text{ and } \dot \Gamma(t) \in \ker(P)\},
\end{equation}
This condition does not uniquely determine the covector $P$, since the form $P$ is defined up to a constant. 
We need to impose a normalization condition on $P$ in such a way to univocally identify it, while $P$ extracts the curvature. Thanks to Proposition \ref{decomposition}, we know that the horizontal component $\dot{\Gamma}\cap({\Ehor} \cap {\chor})$  of a lifted regular curve never vanishes, so that we can impose the following normalization.

\begin{defn}\label{liftM}
  Given  a curve  $\Gamma: I\to M$, lifting of a curve 
    $\gamma: I\to \R^2,$ we call lifting of $\Gamma$
  the curve $G: I \to  \chor^*$ such that for every $t$, 
   $G(t)$ is the unique covector in $S$ such that $G(t)(X_2)=1$ (whose projection in ${\Ever}^*$ is unitary). 
\end{defn}
\begin{rem}
  Note that the normalization $G(t)(X_2)=1$ is an affine condition that depends on the choice of frame $\{X_1,X_2\}$ and therefore on the metric connection.  The appropriate target space for curvature detecting cells should reflect this property.
\end{rem}

 We will show that the lifting in the previous definition is in fact a lifting to a curve in a special affine subbundle of $H^*$. More specifically,  we consider the subset of $H$, $$A=\{X\in H: X-K(X)= X_1\}$$
 which together with the Ehressmann vertical subbundle $\cV\subset H$, and the distinguished non-vanishing section $X_2$  constitute a special affine bundle $\zeta: A\rightarrow M$ over the base manifold $M$ with model vector bundle $\cV\rightarrow M$. We consider the fiber bundle $A^\dagger\rightarrow M$ to be the fiber bundle over $M$ with fiber $A^\dagger_q$ being the set of all affine functions $\phi: A_q\rightarrow \bR$. The dimension of the fiber  $A^\dagger$ is greater by 1 than the dimension of $A_q$. More precisely, the vector bundle $A^\dagger$ is isomorphic to $\cV^*\oplus \bR$.  By 
$1_A$ we will denote the section of $A^\dagger$ which for every $q\in M$ is  the constant function on $A_q$ equal to 1 while by $0_A$ the section of $A^\dagger$ which is constantly equal to 0 . 
\begin{lem}
    There exists a vector bundle isomorphism between $A^\dagger$ and $H^*$.
\end{lem}
\begin{proof}
The global frame $(X_1,X_2)$ of $H$ induces a dual global frame $(X_1^\vee, X_2^\vee)$ of $H^*$, where $X_2^\vee$ is a section of $\cV^*$. We define the bundle map \begin{align*}
\varPhi: H^*\rightarrow A^\dagger
   , (q, b(q) X_1^\vee+a(q)X_2^\vee)\mapsto (q, b(q)0_A+a(q)1_A).
\end{align*} 

The map is fiberwise linear by definition and its inverse is $$\varPhi^{-1}(q, b(q)0_A+a(q)1_A)=(q, b(q)X_1^\vee+a(q)X_2^\vee).$$
\end{proof}

In the category of special affine bundles there is a canonical notion of duality \cite{AffineBundles}. The special affine dual $A^\star$ of the special affine bundle $(A, X_2)$ is an affine
subbundle in $A^\dagger$ that consists of those affine functions on fibers of $A$ the linear
part of which maps $X_2$ to 1. The bundle $A^\star$ is also a special affine bundle with (fiberwise) model vector space, 
$$Vec(A^\star)\simeq\{\phi:A\rightarrow \bR, \phi \text{ affine map s.t.}\phi^L(1_A)=0\},$$
where $\phi^L$ denotes the linear part of the affine map $\phi$.

We are now ready to  show that the lift introduced in Definition \ref{liftM} can be reinterpreted in the language of special affine bundles and consequently that the singed curvature of a planar curve is parametrized by a special affine subbundle of contact horizontal covectors of $M$.
 \begin{thm}\label{CurvatureDetection}
     Let $G: I\rightarrow H^*$ be the lift of a horizontal curve  $\Gamma: I \rightarrow M$ which is the lift of a regular curve $\gamma: I\rightarrow B$, as in Definition \ref{liftM}. Then, the lift $G$ is a section of the special affine subbbundle $\varPhi^{-1}(A^\star)|_\Gamma\rightarrow \Gamma\subset M$. Moreover, if $\gamma:I\rightarrow B$ and $\tilde{\gamma}:I\rightarrow B$ are two planar curves with lifts $\Gamma, G$ and $\tilde{\Gamma}, \tilde{G}$ respectively, then $G(t_0)=\tilde{G}(t_0)$ if and only if $$\gamma(t_0)=\tilde{\gamma}(t_0), ~\dot{\gamma}(t_0)=\dot{\tilde{\gamma}}(t_0),~ \kappa_N\gamma(t_0)=\kappa_N\tilde{\gamma}(t_0).$$
 \end{thm}
\begin{proof}

We start by proving that for every $t$ the lift $G(t)$ lies in a special affine subspace of $H^*_{\Gamma(t)}$. Indeed, by definition $G(t)=(\Gamma(t), b(\Gamma(t))X_1^\vee+X_2^\vee)$ and since $b(\Gamma(t))X_1^\vee+X_2^\vee= \varPhi^{-1}(b(\Gamma(t)) 0_A+ 1_A)$ it follows that $G(t)\in \varPhi^{-1}(A^\star)_{\Gamma(t)}$ for every $t$. Now, for two distinct curves $\gamma$ and $\tilde\gamma$ we assume that for some $t_0$, $G(t_0)=\tilde G(t_0)$. The latter is true if and only if $\Gamma(t_0)=\tilde\Gamma(t_0)\text{ and }P(t_0)=\tilde{P}(t_0)$. From Proposition \ref{decomposition}
 follows that $P(t_0)= \frac{\kappa_N\gamma(t_0)}{\|\dot{\gamma}(t_0)\|}X_1^\vee+ X_2^\vee \text{ and }\tilde P(t_0)= \frac{\kappa_N\tilde\gamma(t_0)}{\|\dot{\gamma}(t_0)\|}X_1^\vee+ X_2^\vee $ and therefore $P(t_0)=\tilde{P}(t_0)$ if and only if $\kappa_N\gamma(t_0)=\kappa_N\tilde\gamma(t_0)$. \end{proof}

 Note that $G (t)$ originally belongs to $H^* \subset T^* M,$ which has dimension 5, but the normalization constraint that we imposed guaranties that  the dimension of the space where we find $G$ is 4. Clearly the normalization is not the only possible one as it depends on the metric connection.
 
Finally, by writing the 1-form $a_{H^*}$ with respect to the dual frame $X_1^\vee$, $X_2^\vee$ and the projection 
$\Pi: T^*(T^*M)\rightarrow T^*M$ $$a_{H^*}=d\Pi^*(\beta(q)X_1^\vee+ a(q) X_2^\vee)$$ we see that the restriction on $\varPhi^{-1}(A^\star)$ is a nowhere vanishing 1-form. Thus, we will denote by $\alpha$ the restriction of the tautological form $\lambda_{T^*M}$  on $\varPhi^{-1}(A^\star)$, 
$$\alpha:= \lambda_{T^*M}|_{\varPhi^{-1}(A^\star)}$$ and we will obtain the 
  the hyperplane   distribution $$\mathfrak{H}:=ker(\alpha)\subset T\varPhi^{-1}(A^\star)$$ on $\varPhi^{-1}(A^\star)$.

   Finally, we show that the distribution $\mathfrak{H}$ characterizes the lifts of \eqref{liftM}.
 \begin{prop} Let $\Gamma: I \rightarrow M$ be the contact horizontal lift of a regular curve $\gamma: I\rightarrow B$ and consider $\tilde{\Gamma}:I\rightarrow  \varPhi^{-1}(A^\star)$ such that $\Pi(\tilde{\Gamma})=\Gamma$, then $\tilde{\Gamma}$ is horizontal with respect to the hyperplane distribution $\mathfrak{H}$ if and only if it satisfies \eqref{liftM}.
 \end{prop}
 \begin{proof}
       Since $\alpha$ is the restriction of the tautological 1-form on $\varPhi^{-1}(A^\star)$ and for $\Pi: T^*(T^*M)\rightarrow T^*M $ we have 
     \begin{align*}
         \alpha_{\tilde{\Gamma}}(\dot{\tilde{\Gamma}})=d\Pi^*(\tilde{\Gamma})(\dot{\tilde{\Gamma}})=\tilde{\Gamma}(d\Pi \dot{\tilde{\Gamma}})=\tilde{\Gamma}(\dot{{\Gamma}}).
     \end{align*}
     Thus, $\tilde{\Gamma}$ is horizontal if and only if $\tilde{\Gamma}(\dot{\Gamma})=0$, which is equivalent to \eqref{liftM}. 
 \end{proof}

To conclude this chapter, for a moment we assume that $(B,g)=(\bR^2,g_\bE)$ and use gauge coordinates $(x,y,\theta) $ for a chart $U\subset M$. At each $H^*_q$, $X_1^\vee$ and $X_2^\vee$ induce linear coordinates $(b,a)$.  The expression of an element of $H^*$ in coordinates $(x,y,\theta, b, a) $ will be 
$$(\gamma(t), a d\theta + b (\cos(\theta) dx + \sin(\theta) dy)).$$
Due to the normalization condition, the expression of the lifting become 
$$(\gamma(t),  d\theta + k (\cos(\theta) dx + \sin(\theta) dy)),$$
where $k$
is the curvature of $\gamma.$
The model reduces to the Curvature Detection Model of Petitot, Citti and Sarti \cite{PetitotCurv} where in that sense the lift selects the curvature parameter $k$.

\section{ Curvature and 
Curvature Radii Space: Engel Structure and Associated Lie Algebra}\label{Radii}
In the previous chapter, we identified the space of Curvature Selective Complex Cells with the special affine bundle $(\varPhi^{-1}(A^\star),X_2^\vee)$.   
 If the curvature does not vanish, we can associate to each curve the  signed radius of the osculating circle.  The radius will be positive if the center belongs to the half line aligned with $N(t)$ at the point $\gamma(t)$, it will be negative otherwise. 
In this paragraph, we will formalize this approach and relate it to the Cartan Prolognation- the projectivization of the contact structure $(M,H^*)$.  We consider this space to represent the joint space of Absolute Curvature and Curvature Radii. Using the canonical Engel structure of the Cartan Prolongation we show that the distribution $\mathfrak{H}$ is Engel (Proposition \ref{Engel Pullback}). We show that the iterative brackets of a pair of generators of the Engel distribution generate the Lie algebra $\mathfrak{sim(2)}$ (Lemma \ref{LieAlgebra} ).
 
\subsection{Cartan Prolongation and Curvature Radii}
 We recall that, the projectivization $\bP H$ of the contact distribution $H\subset TM$
  is the fiber bundle over $M$ with typical fiber $\mathbb{P}_q H$ the lines of $H_q$. The projectivization $\mathbb{P}H$ inherits a canonical 2-dimensional plane distribution defined by declaring that a curve $(q(t),\ell(t))\in M\times\mathbb{P}^1$ is horizontal if and only if the derivative of the point of contact $\dot{q}(t)$ lies on the line $\ell(t)\in H_{q(t)}$ for every $t$.  Equivalently, if $\Pi_{\bP^1}: M\times \mathbb{P}^1\rightarrow M$ is the projection of the fiber bundle and $d\Pi_{\bP^1}: T( M\times \mathbb{P}^1)\rightarrow T( M) $ its differential, the distribution $\mathcal{D}$ is described as follows
\begin{equation}
    \mathcal{D}=\{d{\Pi_{\bP^1}^{-1}}_{(q,\ell)}(\ell_q): (q,\ell) \in  M\times \mathbb{P}^1 \}.
\end{equation}
The projectivization $\mathbb{P}H$ together with the distribution $\mathcal
{D}$ is the \textit{ prolongation} of $M$; for a more comprehensive treatment of prolongations, we refer to \cite{Montgomery}.

In the previous section we showed that there is a unique choice of linear coordinates in the space $H$ induced by $X_1$ and $X_2$. In this paragraph, we highlight the role of these coordinates. The vector fields $X_1$ and $X_2$  form linear coordinates on each contact plane $H_q$ and therefore a line $\ell_q \subset H_q$ can be expressed with respect to these vector fields as
\begin{equation}
    \ell_q= r(s) X_1(q)+ \kappa(s) X_2(q), ~ s>0.
\end{equation}

Using affine coordinates $(r,\kappa)$ on $H$ with respect to the frame $\{X_1,X_2\}$ and gauge coordinates $(x,y,\theta)$ on $M$, the distribution $\mathcal{D}$ is spanned by the vector fields
\begin{equation}\label{rVectors}
   \mathcal{X}= rcos(\theta)\partial_x+rsin(\theta)\partial_y+ \partial_\theta, ~ \mathcal{R}=\partial_r
\end{equation}
on the open submanifold $\mathbb{P}H\setminus \{\kappa = 0\}$ with coordinates $(x,y,\theta,r)$ and by 
\begin{equation}\label{KVectors}
     \mathcal{X}= cos(\theta)\partial_x+sin(\theta)\partial_y+ \kappa\partial_\theta, ~ \mathcal{K}=\partial_\kappa
\end{equation}
on the open submanifold $\mathbb{P}H\setminus \{r= 0\}$ with coordinates $(x,y,\theta,\kappa)$.

The affine chart $\bP H\setminus{\{\kappa=0\}}$ represents the absolute curvature space and the affine chart $\bP H\setminus{\{r=0\}}$ represents the curvature radii space, since in their intersection it holds that  $r=\frac{1}{\kappa}$.

\begin{rem}
    Note that when $B=\bR^2$ the distribution $\cD$ expressed in affine coordinates $(x,y,\theta,\kappa)\in \bP H \setminus\{r > 0\}$ is given by the condition 
    \begin{equation*}
        \dot{q}\in ker(d\theta+k(cos(\theta)dx+sin(\theta)dy))\cap ker(-sin(\theta)dx+cos(\theta)dy)
    \end{equation*}
    which is exactly the condition describing the Engel distribution introduced in \cite{PetitotCurv}.
\end{rem}

We use the projective line bundle $M\times \mathbb{P}^1$ with co-rank 2 distribution $\mathcal{D}$, which is the  Cartan prolongation of the orientation-position contact manifold $(M=\mathbb{S}(T^*B),H)$ to describe the Absolute Curvature-Curvature Radii space.  
One can work similarly for the Cartan prolongation of $(M, H^*)$.

\subsection{Engel Structure}
The prolongation $M\times \bP^1$ of the contact manifold $M$ together with the rank 2-distribution $\mathcal{D}$ is an Engel structure, \cite{Montgomery}. Engel structures form another class of nonintegrable distributions which
is closely related to contact structures. 

By definition, an Engel structure is a
smooth distribution $\mathcal{D}$ of rank 2 on a manifold  of dimension 4 which satisfies
the nonintegrability conditions
\begin{equation}
   rank (\cD+ [\mathcal{D},\mathcal{D}])=3~,~ rank(\cD+[\cD,\cD]+[\mathcal{D},[\mathcal{D},\mathcal{D}]])=4,
\end{equation}
 where $[\mathcal{D}, \mathcal{D}]$ consists of those tangent vectors which can be obtained by taking
commutators of local sections of $\mathcal{D}$. 
 More specifically, for the absolute curvature-position-orientation space $(\mathbb{P}H, \mathcal{D})$ we have on the the open submanifold $\bP M/ \{r = 0\}$
\begin{align*}
 [\mathcal{D},\mathcal{D}
 ]= span\{\mathcal{X},\mathcal{K}, \mathcal{Y}=[\mathcal{X},\mathcal{K}]=\partial_\theta\}\\
 [\mathcal{D},[\mathcal{D},\mathcal{D}]]=span\{\mathcal{X},\mathcal{K},\mathcal{Y}, \mathcal{Z}=[\mathcal{X},\mathcal{Y}]=-sin(\theta)\partial_x+cos(\theta)\partial_y\}.
\end{align*}

We will use the Engle structure of the Cartan Prolongation of $(M,H^*)$ to show that the affine subbundle $\Phi^{-1}(A^\star) $ of $H^*$ modeling the curvature detecting cells together with the tangent distribution obtained by the pullback of the contact 1-form $a_J\in T^*M$ through the projection $\mathcal{P}: T^*M\rightarrow M$ and the hyperplane distribution $\mathfrak{H}$ is an Engel manifold. 
\begin{prop}\label{Engel Pullback}
   The affine subbundle $\Phi^{-1}(A^\star)\subset H^*$ together with  the 2-dimensional plane distribution $\mathfrak{H}\cap ker(\mathcal{P}^*a_J)\subset T\Phi^{-1}(A^\star)$ is an Engel manifold.
\end{prop}

\begin{proof}
Let $\phi$ be the trivialization  map $\bP H^*/\{r=0\}\rightarrow M\times \bR$ $$[q,P=(P_1,P_2)]\mapsto (q, P_1)$$ and let $\psi$ be the  trivialization map of $\Phi^{-1}(A^\star)\rightarrow M\times \bR$ such that $$(q,P=bX_1^\vee+X_2^\vee)\mapsto (q, b)$$, then $F:=\phi \circ \psi^{-1}: \bP H^*/\{r=0\}\rightarrow \Phi^{-1}(A^\star)$ is a diffeomorphism and such that $dF(D)=\mathfrak{H}\cap ker(\mathcal{P}^*a_J).$ Since diffeomorphisms commute with the Lie bracket and $D$ is an Engel distribution, it follows that $\mathfrak{H}\cap ker(\mathcal{P}^*a_J)$ is an Engel distribution.
\end{proof}

In the table bellow we summarize the mechanism of consecutive prolongations from the retina that leads to the projective line bundle $\mathbb{P}H^*$ which model the absolute curvature and scale feature space and how it relates to the affine subbundle of curvature detecting cells.
\[
\begin{tikzcd}[row sep=small]
\begin{array}{c}
 (\Phi^{-1}(A^\star) \subset T^*M ,\mathfrak{H}\cap ker(\mathcal{P}^*a_J))\\
\scriptstyle\text{(signed curvature-position-orientation)}
\end{array}
\arrow[d,"\simeq"]\\
\begin{array}{c}
(\mathbb{P}H^* /\{r=0\},D )\\
\scriptstyle\text{(absolute curvature-curvature radii-position-orientation)}
\end{array}
\arrow[d] \\
\begin{array}{c}
(M,H) \\
\scriptstyle\text{(position-orientation)}
\end{array}
\arrow[d] \\
\begin{array}{c}
B \\
\scriptstyle\text{(position)}
\end{array}
\end{tikzcd}
\quad
\begin{tikzcd}[row sep=large]
\begin{array}{c}
\scriptstyle\text{Affine Subbundle, Engel Structure}
\end{array} \\
\begin{array}{c}
\scriptstyle\text{Cartan Prolongation, Engel Structure}
\end{array} \\
\begin{array}{c}
\scriptstyle\text{Manifold of retinal contact elements , Contact Structure}
\end{array} \\
~
\end{tikzcd}
\]
\begin{rem}
    In regularity theory or more generally in local problems, the important property is that the Engel distribution $\mathcal{D}$ is bracket generating. In fact, locally all Engel structures are equivalent in a similar way that all contact structures are locally equivalent. 
\end{rem}

\subsection{Associated Lie Algebra Structure}
 In this paragraph we examine the behavior of the generators of the Engel distribution $\cD$ under the Lie Bracket. In analogy to the models of orientation detection where the contact distribution $H\subset TM$ is left-invariant with respect to the action of the $SE(2)$ group, we aim to associate a Lie group structure with the Engel structure. Given any element $X$ of a Lie Algebra $\mathfrak{g}$, we shall denote by $ad_X: \mathfrak{g}\rightarrow \mathfrak{g}, ~ad_X(Y)=[X,Y]$ the adjoint representation of $X$ on $\mathfrak{g}$. Moreover, let $X_1,...,X_r$ be $C^\infty$ vector fields on a manifold. Given any non-empty set of integers $I=(i_1,..,i_\ell), 1\leq i_a\leq r \text{ for all } 1\leq a\leq \ell$, we write the multibracket $$X_I= ad_{X_{i_1}}(ad_{X_{i_2}}(...ad_{X_{i_{\ell-1}}})X_{i_\ell})),$$
and we refer to $\ell=|I|$ as the length of the multibracket. We shall denote by $Lie_\bR(X_1,...,X_r)$ the span over $\bR$ of the multibrackets $X_I$ of any length, which together with the Lie bracket is a Lie algebra.  

\begin{lem}
   The Lie algebra $Lie_\bR(\cX, \cK)$ is infinite dimensional.
\end{lem}
\begin{proof}
We use the expression of the generators $\{\cX, \cK\}$ of $\mathcal{D}$ in local coordinates $(x, y, \theta, \kappa)$ on $\bP H /\{r=0\}$ as defined in Equation~\ref{KVectors}. Let us define the vector fields
\[
\mathcal{Z}_1 := -\sin(\theta)\partial_x + \cos(\theta)\partial_y, \qquad \mathcal{Z}_2 := \cos(\theta)\partial_x + \sin(\theta)\partial_y.
\]
These satisfy the following commutation relations:
\begin{align*}
[\mathcal{X}, \mathcal{K}] &= \partial_\theta =: \Theta, \\
[\mathcal{X}, \Theta] &= \mathcal{Z}_1, \\
[\mathcal{X}, \mathcal{Z}_1] &= \kappa \mathcal{Z}_2, \\
[\mathcal{X}, \kappa \mathcal{Z}_2] &= \kappa^2 \mathcal{Z}_1.
\end{align*}
From these, we observe the recursive structure of iterated Lie brackets:
\[
\operatorname{ad}_{\mathcal{X}}^2(\mathcal{Z}_1) = [\mathcal{X}, [\mathcal{X}, \mathcal{Z}_1]] = [\mathcal{X}, \kappa \mathcal{Z}_2] = \kappa^2 \mathcal{Z}_1.
\]
Using induction, we can show that:
\[
\operatorname{ad}_{\mathcal{X}}^{2k}(\mathcal{Z}_1) = \kappa^{2k} \mathcal{Z}_1, \qquad 
\operatorname{ad}_{\mathcal{X}}^{2k+1}(\mathcal{Z}_1) = \kappa^{2k+1} \mathcal{Z}_2 \quad \text{for all } k \geq 0.
\]
Since each of these iterated brackets produces a new vector field involving increasingly higher powers of \(\kappa\), the Lie algebra generated by \(\mathcal{D}\) contains infinitely many $\mathbb{R}-$linearly independent vector fields. Therefore, its Lie closure is infinite-dimensional.
\end{proof}

 We notice that the generators $\{\cX, \cR\}$ of the distribution $\cD$ on $\bP H\setminus{\{\kappa=0\}}$ - the absolute curvature/ curvature radii space - generate a finite dimensional Lie algebra. In the next paragraph, we will use this property to associate locally a finite dimensional Lie group to the space of absolute curvature/ curvature radii.

\subsection{Action of Similitude Group on the Engel structure}\label{Lie GroupStructure}
The goal of this paragraph is to associate a finite Lie group with the prolongation $\mathbb{P}H$ such that the Engel structure $\mathcal{D}$ is left-invariant. We first consider the open submanifold $U\subset\mathbb{P}H$ which is the intersection of the affine charts $$U:=\{(q,\ell)\in \mathbb{S}\tau: \ell_q=rX_1(q)+\kappa X_2(q) \text{ and } \kappa \cdot r \neq 0\}\simeq \mathbb{R}^2\times \mathbb{S}^1\times \mathbb{R}^+.$$ In this open subset of $\mathbb{P}H$, we can consider a new pair of generators for the distribution $\mathcal{D}$, namely in local coordinates $(x,y,\theta, r)$ 
\begin{equation}\label{vectorfields}
    \mathcal{X}_{loc}= \mathcal{X}= r cos(\theta)\partial_x+rsin(\theta)\partial_y+\partial_\theta, ~ \mathcal{R}_{loc}=r\mathcal{R}=r \partial_r.  
\end{equation}

 We recall that the similitude group is the semidirect product $SIM(2)=\mathbb{R}^2\rtimes (SO(2)\times \mathbb{R}^+)$, that is the set $\mathbb{R}^2\times \mathbb{S}^1\times \mathbb{R}^+$ with group law
\begin{equation*}
    ((x,y),\theta,\delta)\cdot ((x^\prime,y^\prime),\theta^\prime, \delta^\prime)=((x,y)+\delta\theta (x^\prime,y^\prime), \theta+\theta^\prime, \delta \delta^\prime).
\end{equation*}
The Lie Algebra  $\mathfrak{sim(2)}$ is generated by the left invariant vector fields \begin{equation}\label{SimAlgebra}
\{\partial_\theta, \delta (\cos(\theta)\partial_x + \sin(\theta)\partial_y ), \delta ( -\sin(\theta)\partial_x + \cos(\theta) \partial_y ), \delta\partial_\delta
\}\end{equation}
The group can be identified with the group generated by planar translations $t_{(x,y)}\in \bR^2$, planar rotations $r_\theta\in SO(2)$ and dilations $d_\delta\in \bR^+$, namely for $(x,y,\theta, \delta)\in SIM(2)$
$$t_{(x,y)}: z\mapsto z+(x,y), ~ r_\theta: z\mapsto r_\theta z  \text{ and } d_\delta: z\mapsto \delta z.$$
Consequently, it acts transitively on $\mathbb{R}^2$ giving rise to the quisi-regular representation on $L^2(\mathbb{R}^2)$ via 
\begin{align}\label{SIMRep}
   \pi(x,y,\theta,\delta)f(z)&=D_\delta R_\theta T_{(x,y)} f(z)
   \\& =|\delta|^{-1} f(d_\delta^{-1}(r_\theta^{-1}t_{(x,y)}^{-1})z)=|\delta|^{-1}f(\delta^{-1}r_{-\theta}(z-(x,y))),~  z\in\mathbb{R}^2.
\end{align}

\begin{rem} 
    Our model contains different aspects of models of scale already present in the literature. Indeed, in 
\cite{Scale}, \cite{Scale2},     
 the  differential constraints are given by a differential 2-form which is derived from the contact 1-form on $SE(2)$ via symplectization. In \cite{Duits}
    the relation between scale and invariances of the $SIM(2)$ group was expoited  and the underlying manifold is $\mathbb{R}^2\times\mathbb{S}^1\times\mathbb{R}^+$. The two models were completely independent, while we understand the relation between the different instruments and the relation with curvature.   
\end{rem}

 In the following Lemma we point out that an identification of the dilation parameter $\delta$ with the curvature radius coordinate $r$ imposes a finite Lie group structure, compatible with the Engel distribution $\cD$,  to the open submanifold $U$ of the absolute curvature/ curvature radii space.
\begin{lem}\label{LieAlgebra}
The Lie algebra $Lie_{\bR}(\cX_{loc}, \cR_{loc})$ is isomorphic to $\mathfrak{sim(2)}$. Moreover, the Engel distribution $\mathcal{D}_U:=\{\mathcal{D}_{(q,\ell)}:(q,\ell)\in U\}$ is left-invariant with respect to $SIM(2)$. 
 
\end{lem}
\begin{proof}
   
    Let us set $\delta= r$. The generators of $\cD_U$, $\mathcal{X}_{loc}$ and $\mathcal{R}_{loc}$, are $\bR-$linear combinations of the left-invariant vector fields in \ref{SimAlgebra}, which correspond to a basis of the Lie algebra $\mathfrak{sim}(2)$.

\end{proof}

To conclude this section, let us display the way the action of the group $SIM(2)$ affects the symmetries of position, orientation and curvature of a planar curve.
Let $\gamma=(x(t),y(t)): \bR\rightarrow \bR^2$ denote a planar curve and let $(x_0,y_0)$ be the point on the curve at a fixed time $t_0$ with non-zero signed curvature $\kappa(\gamma(t_0))\neq0$. We denote by $\theta$ the angle of the tangent vector $\dot{\gamma}(t)$ at $(x_0,y_0)$ with the reference vector field $\partial_x$. We consider the osculating circle of the curve $\gamma$ at the point $(x_0,y_0)$, namely the circle tangent to $\gamma$ at $t_0$ with radius $R_C(t_0)=\frac{1}{\kappa(\gamma(t_0))}$.

 A translation $t_{-(x_0,y_0)}$ and rotation $r_{-\theta}$ translate the curve to the origin of the plane and rotate it such that $\dot{\gamma}(t)$ at the origin forms a $0$ degree angle with $\partial_x$. Consequently, the osculating circle at $(x_0,y_0)$ is now tangent to the horizontal axis at the origin with center $(0,R_C(t_0))$. Finally, a dilation $d_\delta$ scales the radius of the osculating circle $R_C(t_0)\mapsto \delta R_C(t_0)$. In that sense, we can consider the parameter of dilations $\delta$ to be the radius of the osculating circle, 
$
    \delta= \frac{1}{\kappa(t)}
$
and the space of non-zero radii (or inverse curvatures) to be parametrized by $(\mathbb{R}^+,\cdot)$ 

\begin{center}
\begin{figure}
    \centering
    \begin{tikzpicture}[scale=0.9]


\begin{scope}[shift={(2,0)}]
\node at (1,-1.2) {};
\draw[->] (-1.2,0) -- (1.8,0) node[below right] {$x'$};
\draw[->] (0,-0.5) -- (0,2) node[left] {$y'$};

\draw[thick, domain=-1.2:1.5, smooth, variable=\x] 
  plot({\x}, {0.2*\x*\x});

\filldraw[black] (0, 0) circle (0.03);
\node[below left] at (0, 0) {$(0,0)$};

\draw[->, thick, red] (0,0) -- (1,0);
\node[above] at (0.5,0.05) {$\dot{\gamma}(t)$};

\draw[blue, dashed] (0,0.75) circle (0.75);
\node[blue] at (1.4,0.75) {\small osculating circle};
\end{scope}

\begin{scope}[shift={(8,0)}]
\node at (1,-1.2) {};
\draw[->] (-1.2,0) -- (1.8,0) node[below right] {$x$};
\draw[->] (0,-0.5) -- (0,2.5) node[left] {$y$};
\draw[thick, domain=-1.2:1.5, smooth, variable=\x] 
  plot({\x}, {0.2*\x*\x});
\draw[thick, domain=-1.2:1.5, smooth, variable=\x] 
  plot({\x}, {0.6*\x*\x});
\draw[thick, domain=-1.2:1.5, smooth, variable=\x] 
  plot({\x}, {0.9*\x*\x});
\draw[thick, domain=-1.2:1.5, smooth, variable=\x] 
  plot({\x}, {1.2*\x*\x});
\foreach \r in {0.3, 0.6, 0.9, 1.2} {
  \draw[blue, dashed] (0,\r) circle (\r);
}

\node[blue] at (1.4, 1.2) {\small osculating circle family};

\filldraw[black] (0, 0) circle (0.03);
\node[below left] at (0, 0) {$(0,0)$};

\end{scope}

\end{tikzpicture}
    \caption{The Frenet approximation of a planar curve $\gamma$, shows that the best quadratic approximation is a parabola which expressed in the $(x,y)$ plane, has the shape $y=\kappa \frac{x^2}{2}$ near $\gamma(0)$ where $\kappa$ is its signed curvature. When the curve is not flat, the parabola is approximated by the osculating circle at $\gamma(0)$ \cite{Parent}. The radius of the osculating circle is the scale parameter $\delta$ of the Lie group $SIM(2)$. }
    \label{fig:CurvatureRadii}
\end{figure}
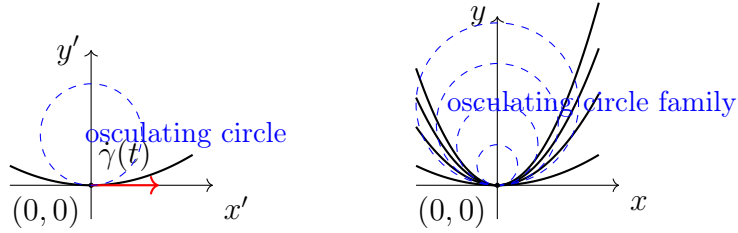

\end{center}

\section{Receptive Profiles of Curvature Selective Cells}
\label{ReceptiveProfs}
The construction of Chapter 3 identifies the space of curvature-selective cells and characterizes the curves compatible with cortical connectivity, but it does not yet describe how individual cells respond to a visual stimulus. This chapter addresses that question by constructing a family of curvature-sensitive receptive profiles extending the Gabor-type filters of the orientation layer to a two-layer architecture.

The construction is guided by two axioms. Hierarchical composability requires that the filters be built in successive layers from orientation-selective filters, reflecting the modular organization of the visual cortex \cite{CortArch}. Minimal uncertainty requires that they minimize a generalized uncertainty principle with respect to the differential operators induced by the $SIM(2)$ geometry identified in the previous chapter \cite{Uncertainty}. We show that both axioms are simultaneously satisfied by exploiting the $SIM(2)$ group structure. In section \ref{SIMTransform} we define a two-layer integral transform  composing the $SE(2)-$transform with an integral transform on $L^2(SE(2))$ and define accordingly the family of curvature receptive profiles (Definition \ref{CurvRecProf}). We show that if the mother wavelet of the first layer -that of the $SE(2)-$transform- satisfies a covariance relation for $SIM(2)$, then the two-layer transform coincides with the $SIM(2)-$transform (Proposition \ref{SIMTr}). Finally, in Section \ref{Charact} we use a $SIM(2)-$adapted uncertainty principle to characterize the curvature receptive profiles realized as the coherent states of this transform (Theorem \ref{CharRec}).

\subsection{Modeling Curvature-Selective Receptive Fields via Two-Layer Wavelet Transforms}\label{SIMTransform}
In this section, we aim to combine the $SE(2)$ transform which models the functionality of orientation sensitive simple cells and the $SIM(2)$ transform- previously used to model the functionality of scale detection to obtain a continuous wavelet transform adequate to model curvature sensitive cells.

 \paragraph{\textbf{Hierarchical Composability}}
 
Curvature processing operates on the output of orientation-selective cells. We recall that the Lie group $SE(2)$ parametrizes the space of simple cells sensitive to orientation, namely given the mother profile $\Psi_0\in L^2(\bR^2)$ introduced in \ref{MotherProfile}  the family of receptive profiles associated to orientation is
$$RP_{SE(2)}=\{\rho(x,y,\theta)\Psi_0: (x,y,\theta)\in SE(2)\}.$$
where $\rho: SE(2)\rightarrow \cU(L^2(\bR^2))$ is the quasi-regular representation \ref{SERep} of $SE(2)$ on $L^2(\bR^2)$. Curvature processing relies on processing the output of  position and orientation detecting cells which is the output $\cO$ of the SE(2)-transform \ref{SETransform} on a stimulus $I\in L^2(\bR^2)$,
$$
    \cO_{\Psi_0}(x,y,\theta)(I)=\langle I, \rho(x,y,\theta)\Psi_0\rangle_{L^2(\bR^2)}.
$$
A curvature-selective cell receives this response as its input. To model this, we need a linear operator on the range $\cH\subset L^2(SE(2))$ of $SE(2)$-transform which maps the response $\cO(x,y,\theta)$ of a cell sensitive to position $(x,y)$ and orientation $\theta$ to the response of a cell sensitive to position, orientation and curvature radius $(\bar{x},\bar{y},\bar{\theta},\bar{\delta})\in SIM(2)$, namely a linear operator $$\cK: \cH\rightarrow L^2(SIM(2)).$$
To construct $\mathcal{K}$, we embed both 
both $SIM(2)$ and $SE(2)$ as subgroups of the affine group
$$
\mathrm{Aff}(\mathbb{R}^2)
= \left\{(A,t) : \mathbf{z}\mapsto A\mathbf{z}+t,\;
A\in GL(2,\mathbb{R}),\; t\in\mathbb{R}^2 \right\},
$$
equipped with the group law
$$
(A,t)(B,s) = (AB,\; As + t).
$$

where
\begin{align*}
SE(2)
= \{(r_\theta,t) : r_\theta\in SO(2),\ t\in\mathbb{R}^2\},
~\text{ and }\\
SIM(2)
= \{(\delta r_\theta,t) : r_\theta\in SO(2),\ \delta>0,\ t\in\mathbb{R}^2\}.
\end{align*}
The inclusion $SE(2)\subset SIM(2)$ makes $SE(2)$ a homogeneous space of $\text{SIM}(2)$
: the group $\text{SIM}(2)$
 acts on $SE(2)$
 by conjugation inside $Aff(\bR^2)$.
Explicitly, for $h=(\delta r_\theta,t)\in SIM(2)$ and $g=(r_{\bar\theta},\bar t)\in SE(2)$, the action map is 
$$ (s,e)\mapsto h\diamond g:=
h g h^{-1}
= \big( r_\theta r_{\bar\theta} r_\theta^{-1},\;
\delta r_\theta \bar t + t - r_\theta r_{\bar\theta} r_\theta^{-1} t \big).
$$
 This conjugation induces the quasi-regular representation \begin{align*}\tilde\pi: SIM(2)\rightarrow ~ &  \cU(L^2(SE(2))),\\ &\pi(\delta r_\theta,t)f(r_{\bar{\theta}}, \bar{t})= \delta^{-1}f( r_{\theta}^{-1}r_{\bar{\theta}}r_{\theta},\delta r_{\theta}^{-1}(\bar{ t})+t-r_\theta^{-1}r_{\bar{\theta}} r_\theta (t))
\end{align*}
which is unitary with respect to the Haar measure on $SE(2)$ and will serve as the building block for $\mathcal{K}$ in the next section.

In order to state the main result of this section we need to recall the notion of admissibility of a vector  in a Hilbert space with respect to a unitary representation of a locally compact group $G$.
\begin{defn}{\cite{Fuhr}}\label{DefAdm}
    Let $(r, H_r)$ denote a strongly continuous unitary representation of a locally compact group $G$. A vector $\phi\in H$ is called \textit{admissible} iff the operator $H\rightarrow L^2(G), ~ f\mapsto \langle f, r(g)\phi\rangle$ is an isometry.
\end{defn}

We are now ready to define a family of orientation, position, and curvature sensitive receptive profiles as functions on \( SE(2) \), that reflects the hierarchical structure of the visual cortex.
\begin{defn}\label{CurvRecProf} Given an admissible \emph{mother window} \( \tilde{\Psi}_0 \in L^2(SE(2)) \),
    the orientation-position-curvature receptive profiles on $SE(2)$ are the family
\begin{equation}
    RP_{Curvature}=\{\tilde{\Psi}_{h}:= \tilde{\pi}(h)\tilde{\Psi}_0:SE(2)\rightarrow \bR: h\in SIM(2)\}
\end{equation} and the response of this family of receptive fields is given by the operator 
 \begin{align}
 \begin{split}
          \cK_{\tilde{\Psi}}: L^2(SE(2))\rightarrow L^2(SIM(2))\\
   \cO(g)\mapsto  \cK{(h)}(\cO)&=\langle\mathcal{O},\tilde{\pi}(h)\tilde{\Psi}_0\rangle\\
   &=\int_{SE(2)}\cO(g)\overline{\tilde{\Psi}_{h}(g)}d\mu_{g}
 \end{split}
 \end{align}

where $d\mu_{g}$ denotes the Haar measure on $SE(2)$.
\end{defn}
  
 We will show that one can reinterpret the operator $\cK$ as a wavelet transform on the group $SIM(2)$, if $\tilde{\Psi}_{0}$ satisfies some sufficient conditions. First, we introduce some notation in order to differentiate group multiplication and conjugation: the action of $h\in SIM(2)$ on $ X\in \mathbb{R}^2$ will be denoted by 
 $$X\mapsto h\bullet X.$$
 
 Based on the early work by Antoine and Murenzi \cite{AntoineMurenzi}, Sharma and Duits \cite{Duits} the generalized continuous wavelet transform on the $SIM(2)$ group, $\mathcal{W}_\Phi I: SIM(2)\rightarrow $ of a function $I\in L^2(\bR^2)$ is obtained by means of a convolution kernel $\Phi: \mathbb{R}\rightarrow \bC$ via 
 \begin{equation}
     \mathcal{W}_\Phi I(h)=\int_{\bR^2}I(X)\overline{\Phi(h^{-1}\bullet X)}dX, ~ h\in SIM(2).
 \end{equation}
 Assuming $\Phi\in L^2(\bR^2)$ the transform can be rewritten as $\mathcal{W}_\Phi I(h)=\langle \pi(h)\Phi, f\rangle_{L^2(\bR^2)},$ where $\pi$ denotes the quasi-regular representation of $SIM(2)$ on $L^2(\bR^2)$ (see \ref{SIMRep}).  Moreover, a vector $\Phi\in L^2(\bR^2)$ is called admissible\cite{AntoineMurenzi} if \begin{equation}\label{admissibility}
    \int_{\bR^2} \frac{|\hat{\Phi}(w)|^2}{|w|^2}dw=1.
\end{equation} 
\begin{defn}
    We say that a $\Psi_0\in L^2(\bR^2)$ is $SIM(2)-$covariant if there exists some $F:\bR^+\rightarrow \bR^*$ such that \begin{equation}\label{SIMCovariance}
        \pi(h)\Psi_0(-)=F(\delta)\Psi(h^{-1}\bullet -),
    \end{equation}
    for every $h=(x,y,\theta, \delta)\in SIM(2).$
\end{defn}
The following proposition is the central result of this section. It shows that under natural conditions on the mother profiles, the two-layer operator $\mathcal{K}$
 is equivalent to the generalized continuous wavelet transform on $\text{SIM}(2)$
 studied in \cite{AntoineMurenzi, Duits}. This equivalence is what connects the hierarchical cortical model to the existing harmonic analysis literature.

\begin{prop}\label{SIMTr} Let $\Psi_0\in L^2(\bR^2)$ satisfying condition \ref{Calderon} and $\tilde{\Psi}_0\in L^2(SE(2))$  admissible in the sense of \ref{DefAdm}. If $\Psi_0$ satisfies the $SIM(2)$ covariance relation \ref{SIMCovariance}, then for $h\in SIM(2)$ and $X\in \bR^2$
\begin{align*}
    \cI_h(X):= \int_{SE(2)}\overline{\rho(g)\Psi_0(X)\tilde{\pi}(h)\tilde{\Psi}_0(g)}d\mu_g = \pi(h)f(X)\\
    \text{ where } f(X):= \int_{SE(2)}\overline{\Psi_0(g^{-1}\bullet X)\tilde{\Psi}_0(g)}d\mu_g.
\end{align*}
Moreover, if $f$ satisfies the admissibility condition \ref{admissibility},
then $\cK_{\tilde{\Psi}_0}(\cO_{\Psi_0}I)= F(\delta)\cW_{f}I$. In particular, when $F=\delta^{-1}$ it is the generalized continuous wavelet transform associated to the quasi-regular representation $\pi: SIM(2)\rightarrow \cU(L^2(\bR^2))$. 

\end{prop}

\begin{proof} Let $h=(x,y,\theta,\delta)\in SIM(2)$ and $X\in \bR^2$, then we have the following
\begin{align*}
    \cI_h(X)=\int_{SE(2)}\overline{\Psi_0(g^{-1}\bullet X)\delta^{-1}\tilde{\Psi}_0(h\diamond g)}d\mu_g \overset{g^\prime= h\diamond g}{=} \delta\int_{SE(2)}\overline{\Psi_0(h(g^{\prime})^{-1}h^{-1}X)\tilde{\Psi_0}(g^\prime)}d\mu_{g^\prime}.
\end{align*}
Applying the covariance relation, we obtain
\begin{align*}
    \cI_h(X)= F(\delta) \delta \int_{SE(2)}\overline{\Psi_0((g^\prime)^{-1}h^{-1}X)\tilde{\Psi}_0(g^\prime)}d\mu_{g^\prime} =F(\delta)\delta (\pi(h)f(X)).
\end{align*}
Now, we can write
    \begin{align*}
    \cK_{\Psi_0}{(h)}(\cO)=\int_{SE(2)}\cO(g)\overline{\tilde{\pi}(h)\tilde{\Psi}_0(g)}d\mu_g&=\int_{SE(2)}\langle I,\rho(g)\Psi_0 \rangle\overline{\tilde{\pi}(h)\tilde{\Psi}_0(g)}d\mu_g\\
    &= \int_{\bR^2} \int_{SE(2)}I(X)\overline{\rho(g)\Psi_0(X)}\overline{\tilde{\pi}(h)\tilde{\Psi}_{0}(g)}d\mu_gdX\\
    &= \int_{\bR^2}I(X)\underbrace{\int_{SE(2)}\overline{\rho(g)\Psi_{0}(X)}\overline{\tilde{\pi}(h)\tilde{\Psi}_{0}(g)}d\mu_g}_{\cI_h(X)}dX\\
    &=\int_{\bR^2}I(X)F(\delta)\delta\overline{\pi(h)f(X)}dX= F(\delta)\delta\cW_fI
\end{align*}
which is equal to the $SIM(2)$-wavelet transform $\cW_fI$ exactly when $F(\delta)=\delta^{-1}.$

\end{proof}
We consider $\pi(h)f$ as \textbf{the effective receptive profile} of a second-layer cell sensitive to an element $h\in SIM(2)$.
\begin{rem}
Although $\mathcal{K}$
 is constructed in two steps, the proposition shows that under the stated conditions on $\Psi_0$ and $\tilde{\Psi}_0$
 it collapses to a single $SIM(2)$-wavelet transform. This identification gives direct access to the invertibility and sampling theory developed in literature (see for instance \cite{Duits}), where in particular square-integrability of the representation $h \mapsto \pi(h)$ and admissibility of $f$
 in the sense of \eqref{admissibility} are shown to be sufficient but not necessary. Determining the minimal conditions on $\mathcal{K}$ under which the full sampling theory applies remains an open question.
\end{rem}
    \begin{figure}
     \centering
     \begin{tikzcd}
         Contrast \arrow{r} & Orientation \arrow{r} &  Curvature
     \end{tikzcd}
     \begin{tikzcd}
    PW_R \arrow{r}{SE(2)}& \cH\subset L^2(SE(2)) \arrow{r}{}&  L^2(SIM(2)) \arrow[from=1-1,to=1-3, bend left=40]{r}{SIM(2)}
\end{tikzcd}
 
     \caption{The two-layer detection hierarchy. A visual stimulus (contrast) is first processed by orientation-selective cells via the $SE(2)$
-transform, producing a response in $H\subset L^2(SE(2))$. The operator $\mathcal{K}$
 then maps this response to $L^2(\text{SIM}(2))$, where curvature-selective cells read out position, orientation, and curvature  simultaneously. The composite map from contrast to curvature is equivalent to the $\text{SIM}(2)$
-wavelet transform with effective kernel $f$, when the orientation mother profile satisfies the $SIM(2)$ covariance.
}
     \label{fig:Wavelets}
 \end{figure}

\subsection{Characterization of Receptive Profiles as Minimizers of an Uncertainty Principle}\label{Charact}

In Section~\ref{Lie GroupStructure} we showed that, away from the singularity  $\kappa=0$, the vector fields spanning the Engel distribution are left-invariant with respect to the group $SIM(2)$. In Section~\ref{SIMTransform} we further established that, under suitable constraints, the response of curvature-sensitive cells can be expressed as a $SIM(2)$ transform with effective receptive profiles of the form
\[
\{\pi(h)f\}_{h\in SIM(2)} \subset L^2(\mathbb{R}^2).
\] 

We now characterize the shape of these effective receptive profiles under the assumption that they detect features optimally, in the sense of satisfying an uncertainty principle adapted to the curvature feature space. The same mechanism underlies both orientation detection and curvature detection.
 
 The curvature feature space is associated to the Lie group $SIM(2)$ acting on $\mathbb{R}^2$ through the quasi-regular representation $\pi$ ~\eqref{SIMRep}. Since $\pi$ is unitary, for each $X\in \mathfrak{sim}(2)$, the differential $d\pi(X)$ is the skew-adjoint operator on $L^2(\mathbb{R}^2)$:
\begin{align*}
d\pi : \mathfrak{sim}(2) &\longrightarrow \mathrm{SkewAdj}(L^2(\mathbb{R}^2)), \\
X &\longmapsto d\pi(X)f
= \left.\frac{d}{dt}\right|_{t=0} \pi(\exp(tX))f,
\qquad f\in L^2(\mathbb{R}^2).
\end{align*}

In particular, the infinitesimal roto-translation and dilation generators 
$\cX_{loc}$ and $\cR_{loc}$ give rise to the skew-adjoint operators
\[
Y_1 := d\pi(\cX_{loc}), 
\qquad 
Y_2 := d\pi(\cR_{loc}),
\qquad 
Y_3 := d\pi([\cX_{loc},\cR_{loc}]),
\]
which act on the effective receptive profiles.
\begin{lem}\label{DifOper}
Let $(\delta,\theta,x_0,y_0)\in SIM(2)$ and let $(\xi,\eta)$ denote the transformed coordinates in $\mathbb{R}^2$ defined by
\[
(\xi,\eta)=(\delta,\theta,x_0,y_0)^{-1}\bullet (x,y).
\]
In these coordinates, the skew-adjoint operators $X_1,X_2,X_3$ take the form
\[
Y_1 = -\eta \partial_\xi + \xi \partial_\eta + \partial_\xi,
\qquad
Y_2 = \xi \partial_\xi + \eta \partial_\eta,
\qquad
Y_3 = \partial_\eta.
\]
\end{lem}

\begin{proof}
    We determine $Y_1=d\pi(\cX_{loc})$ from the intertwining relation $$\cX_{loc}\pi(\delta,\theta,x_0,y_0)u=\pi(\delta,\theta,x_0,y_0)(Y_1u)$$ for every $u \in L^2(\bR^2)$. 
    Recall that the representation is given by
\[
\pi(g)u(x,y)
= \delta^{-1}
u\!\left(\delta^{-1} r_{-\theta}(x-x_0,y-y_0)\right).
\]
Introduce the transformed coordinates
\[
(\xi,\eta)
= \delta^{-1} r_{-\theta}(x-x_0,y-y_0).
\]
Then
\[
\pi(g)u(x,y)=\delta^{-1} u(\xi,\eta).
\]

\medskip
    
Define the differential operator $Y_1=-\eta\partial\xi+\xi\partial_\eta+\partial_\xi$. A direct calculation gives 
    \begin{align*}
        \cX_{loc}\pi(\delta,\theta,x_0,y_0)u(x,y)&= (\delta cos(\theta)\partial_x+\delta sin(\theta)\partial_y+\partial_\theta) \delta^{-1}u(\delta^{-1}r_{-\theta}(x-x_0,y_y-y_0))\\
        &=\delta^{-1}\partial_{\xi}u(\xi,\eta)+ \delta^{-1}\eta \partial_\xi u(\xi,\eta)-\delta^{-1}\xi\partial_\eta u(\xi,\eta)\\
        &=\delta^{-1}(Y_1u(\xi,\eta))=\pi(\delta,\theta,x_0,y_0)((Y_1u)(x,y)).       
    \end{align*}
    Through the same procedure, we compute $\cR_{loc}\pi(\delta,\theta,x_0,y_0)u=\pi(\delta,\theta,x_0,y_0)[(\eta\partial_\eta+\xi\partial_\xi) u]=\pi(\delta,\theta,x_0,y_0)Y_2 u$ and $[\cX_{loc},\cR_{loc}]\pi(\delta,\theta,x_0,y_0)u=Y_3 \pi(\delta,\theta,x_0,y_0)u$.
\end{proof}

Following the methodology applied on receptive profiles for orientation sensitive cells \cite{CortArch}, the shape of a curvature receptive profile is determined by the  uncertainty principle for connected Lie groups 
 (see Theorem 2.4,\cite{FollandUncertainty}) in terms of the vector fields $X_1, X_2$ and $X_3$,

\begin{equation}\label{SIMUncertaintyEq}
    \|X_1 u\|_{L^2(\bR^2)} \|X_2 u\|_{L^2(\bR^2)} \geq \frac{1}{2}
 |\langle X_3 u, u\rangle _{L^2 (\bR^2)}|.
\end{equation}

 This is an uncertainty principle for certain generators of the $sim(2)$ algebra, while a full uncertainty principle on the group is yield by three operators and four non-zero commutators,
which generate in turn a system of four differential equations.  It has been shown that there is no non-trivial canonical function
which minimizes the full uncertainty system associated with the similitude group for the operators corresponding to rotation, translation and dilation \cite{AntoineMurenzi}.  Instead, solutions have been obtained by considering operators in the enveloping algebra \cite{Dahlke} or using the symmetries in the set of commutators of $\mathfrak{sim}(2)$\cite{AliAntoine}. 

In the following theorem we characterize the Effective receptive profiles of curvature sensitive cells as minimizers of only one non-zero commutator- that of roto-translation and dilation- which is the one obtained by the generators of the Engel distribution.

\begin{thm}\label{CharRec}
Let $u\in L^2(\bR^2)$ and $u\in \{\pi(h)f\}_{h\in SIM(2)}$. Then $u$ should satisfy the first-order partial differential equation

\begin{equation}
    Lu=0, ~ \lim_{\|(\xi,\eta)\|\rightarrow\infty}u=0
\end{equation}
where $L$ is the complex vector field $$L=\big(\eta+1+i\lambda \xi\big)\partial_\xi+(-\eta+i\lambda \xi)\partial_\eta~,~ \alpha,b: \Omega\rightarrow \bC \text{ and }\lambda\in \bR^*.$$

\end{thm}

\begin{proof}
    The uncertainty principle \ref{SIMUncertaintyEq} is minimized for $u\in L^2(\bR^2)$ satisfying the uncertainty minimizer equation
\begin{equation}
    (X_2 -i\lambda X_1) u = 0
\end{equation}
for some $ \lambda \in \bR$. Via Expressing $Y_1$ and $Y_2$ in coordinates $(\xi,\eta)$ as in Lemma \ref{DifOper} and via variable separation we obtain
\begin{align*}
(X_2 -i\lambda X_1) u& =\alpha(\xi, \eta)\partial_\xi+b(\xi,\eta)\partial_\eta
    \\&=\big(\eta+1+i\lambda \xi\big)\partial_\xi+(-\eta+i\lambda \xi)\partial_\eta= Lu.
\end{align*} 
Moreover, $L$ is a non-vanishing complex vector field when $\lambda=0$, since $L$ vanishes only when $\lambda=0$ at $(\xi,\eta)=(0,1)$.
\end{proof}
 
\begin{acknowledgements}
  We would like to thank Nicola Arcozzi, Mattia Galeotti and Alessandro Sarti for the helpful discussions. 
\end{acknowledgements}

\bibliographystyle{acm} \bibliography{Curvature}

\end{document}